\documentclass[letterpaper,journal]{IEEEtran}
\usepackage{amsmath,amsfonts,amssymb}
\usepackage{amsthm}
\usepackage{algorithmic}
\usepackage{multirow}
\usepackage{booktabs}
\usepackage[ruled,linesnumbered]{algorithm2e}
\usepackage{array}
\usepackage[caption=false,font=footnotesize,labelfont=rm,textfont=rm]{subfig}
\usepackage{textcomp}
\usepackage{stfloats}
\usepackage{url}
\usepackage{verbatim}
\usepackage{graphicx}
\usepackage[colorlinks,citecolor=blue,linkcolor=blue,urlcolor=blue]{hyperref}
\usepackage{cite}
\usepackage{xcolor}
\newcommand{\bm}[1]{\boldsymbol{#1}}
\newcommand{\tr}{\operatorname{tr}}
\newcommand{\diag}{\operatorname{diag}}
\newcommand{\vecop}{\operatorname{vec}}

\newtheorem{theorem}{Theorem}

\newtheorem{proposition}{Proposition}
\newtheorem{corollary}{Corollary}
\newtheorem{remark}{Remark}

\begin{document}

\title{Reconfiguring Sparse Apertures: Partial-Mobility Planar FAS for 2-D DOA Estimation}

\author{Xiaokai Song,
	Tuo Wu,
	Jie Tang, 
	Maged Elkashlan,  
	K. C. Ho,~\IEEEmembership{Fellow,~IEEE},
	Kin-Fai Tong,~\IEEEmembership{Fellow,~IEEE} 
	
	\thanks{Xiaokai Song, Jie Tang, and Tuo Wu are with the School of Electronic and Information Engineering, South China University of Technology, Guangzhou 510640, China (e-mail: $\rm  songxiaokai@scut.edu.cn; wutuo@scut.edu.cn; eejtang@scut.edu.cn  $).}
	\thanks{Maged Elkashlan is with the School of Electronic Engineering and Computer Science, Queen Mary University of London, London E1 4NS, U.K (e-mail: $\rm maged.elkashlan@qmul.ac.uk$).} 
    \thanks{K. C. Ho is with the Department of Electrical Engineering and Computer Science, University of Missouri, Columbia, MO 65211, USA (e-mail: $\rm hod@missouri.edu$).}
    \thanks{K.-F. Tong is with the School of Science and Technology, Hong Kong Metropolitan University, Hong Kong, China (e-mail:$\rm ktong@hkmu.edu.hk$).} 
}

\maketitle

\begin{abstract}
Classical sparse arrays enlarge the sensing aperture under limited element and radio-frequency (RF) chain budgets, but their fixed geometries impose a persistent tradeoff between wide-sector identifiability and local angular resolution. This paper converts this static design into a two-state reconfigurable sparse aperture by allowing only part of a planar fluid antenna system (FAS) to move after an initial observation. A compact, Nyquist-spaced seed first provides an ambiguity-controlled acquisition geometry. The remaining fluid ports then move or switch to information-efficient refinement positions, and the measurements collected before and after reconfiguration are jointly processed. We formulate the receiver under fixed total numbers of ports and RF chains, a movement budget, and a common total-snapshot budget. A policy-level Fisher information matrix (FIM) identity accounts for the observation-dependent refinement geometry. For a single source, a relaxed planar D-optimal analysis gives the corner-favoring aperture law, while a finite-port certificate bounds the information retained under spacing, reachability, and partial-actuation constraints. For multiple sources, an aperture-and-conditioning surrogate generates feasible sparse layouts that are ranked and locally refined using the exact frame FIM with coarray, spacing, and movement regularization. Seed multiple signal classification (MUSIC) identifies a reliable angular basin, followed by joint concentrated-likelihood refinement over both states. Equal-budget simulations show that reconfiguring a sparse aperture converts available area and movement into lower angular error and retains much of the all-movable gain with fewer actuated ports. They also identify the operating boundary: if a full sparse aperture can remain permanently deployed, avoiding movement and snapshot splitting can be preferable.
\end{abstract}



\section{Introduction}\label{sec:intro}

\IEEEPARstart{S}{parse} arrays are a classical response to a fundamental direction-finding constraint: high angular resolution favors a large physical aperture, whereas the available antenna elements, radio-frequency (RF) chains, feeds, and calibration resources are limited. Minimum-redundancy, nested, and coprime arrays spread a finite number of sensors over a larger support or enlarge the difference-coarray degrees of freedom~\cite{moffet_mra,pal_nested,vaidyanathan_coprime}. Compressed sparse arrays retain the aperture of a larger fixed linear array with fewer RF chains~\cite{guo_compressed_sparse}, while two-dimensional extensions use gridless recovery, sparse billboard or T-shaped geometries, and tensor processing for nested planar arrays~\cite{wang_fri_2d,alawsh_billboard,xu_tensor_lnested}. These developments establish sparse geometry as an effective way to exchange dense spatial sampling for aperture and information.

Nonetheless, the same sparsity that increases resolution also creates a static geometry tradeoff. Long and nonuniform baselines improve local steering-manifold sensitivity, but may produce narrow attraction basins, elevated sidelobes, or grating-lobe ambiguity. A compact half-wavelength array offers a more reliable wide-sector search, but sacrifices the aperture needed to separate nearby sources. Moreover, a fixed sparse layout is optimized once for a prescribed sector or source-separation regime and cannot reallocate its limited sensors after observing the scene. Classical array theory can characterize either a compact acquisition array or a high-resolution sparse array, but it cannot make the same hardware change roles within a sensing frame. This motivates a sharper question: can a finite set of physical ports first provide ambiguity-controlled acquisition and then reconfigure into a task-adapted sparse aperture for refinement?

Fluid antenna systems (FASs) provide the missing operational degree of freedom by making the active antenna position controllable after deployment~\cite{wong_fas,wong_fas_prelim}. Liquid-metal motion, mechanical translation, calibrated pixel or port switching, and electronically reconfigurable meta-fluid surfaces offer different physical routes to this reconfigurability~\cite{new_fas_tutorial,pixel_fas,wu_scalable_fas,Meta-Fluid lot}. Hardware surveys and 6G perspectives therefore identify position and electromagnetic reconfiguration as native FAS capabilities~\cite{wu_fas_6g,new_fas_redefining}. Most algorithmic FAS studies, however, optimize communication objectives such as multiple access~\cite{wong_fama,slow_fama,wu_fas_conoma}, transmission with statistical channel state information~\cite{fas_stat_csi}, or physical-layer security~\cite{wu_fas_secrecy}. Direction finding follows a different mechanism: steering-manifold derivatives determine local Fisher information, whereas the global ambiguity or likelihood surface determines whether an estimator enters the correct angular basin~\cite{van_trees,stoica_music,kay}. Communication-oriented placement objectives therefore do not directly determine which ports should move, when reconfiguration should occur, or how observations collected on both sides of the movement should be combined.

Recent sensing studies begin to connect mobility and sparse-array processing. They include one-dimensional sparse FASs with aligned or misaligned observations~\cite{xu_future_fluid}, time-constrained designs with fully movable or fixed-reference structures~\cite{xu_time_constrained_fas}, scalable compact-to-expanded configurations~\cite{wu_scalable_fas}, continuous D-optimal sparse-FAS placement~\cite{wu_sparse_fas_continuous}, hybrid analog--digital FASs for compressive 2-D DOA estimation~\cite{tian_hybrid_fas_doa}, and sparse-array synthesis through arbitrary linear motion~\cite{zhang_moving_2d}. Related movable-antenna work addresses MIMO capacity~\cite{ma_mimo_capacity}, multiuser transmission~\cite{zhu_ma_multiuser}, antenna-position optimization~\cite{mei_graph}, channel estimation~\cite{xiao_ma_channel}, prototype-validated position control~\cite{dong_ma_proto}, and sensing or integrated sensing and communications~\cite{ma_movable_sensing,lyu_ma_isac}. Meta-fluid architectures have also been developed for multi-user ISAC with electronic reconfiguration and full-wave validation~\cite{Meta-Fluid}.
These results confirm that antenna geometry and electromagnetic states can be optimized, but they do not establish a partially actuated planar sparse-aperture protocol with a fixed acquisition seed, an observation-dependent refinement state, joint processing of pre- and post-movement measurements, explicit movement and settling costs, and matched RF-chain and snapshot budgets.

The objective of this work is therefore not to synthesize another static sparse array, but to turn sparse-aperture design into a causal receiver operation. The same bounded hardware first uses a compact, calibrated, half-wavelength seed for wide-sector acquisition. After the seed observation identifies where additional resolution is needed, only the remaining fluid ports move or switch to create long, source-conditioned sparse baselines. Measurements from both states are retained. Moving every port would enlarge the feasible placement set, but would also increase actuation, switching, recalibration, and settling overhead~\cite{new_fas_tutorial,xu_time_constrained_fas}. Partial mobility instead asks how much of the all-movable information can be retained per actuated port while preserving the acquisition geometry.

This transition from a static array to an adaptive sparse aperture creates four coupled challenges. First, the refinement geometry is a function of the seed data, so the information analysis must describe a data-adaptive experiment and accumulate rather than replace the information acquired before movement. Second, the moved sparse aperture must increase local sensitivity without making the global search unreliable; for multiple sources, coordinate spread alone does not ensure a well-conditioned steering manifold or adequate separation~\cite{van_trees,stoica_music,wu_scalable_fas}. Third, the compact seed may not resolve closely spaced sources at low signal-to-noise ratio, so its confidence must control both whether movement is triggered and which angular region guides the refinement layout. Fixing this seed also shrinks the placement set, making the retained information a quantity to be bounded rather than assumed. Fourth, a realizable policy must respect travel, spacing, settling, RF-access, and total-observation constraints, while comparisons must distinguish equal snapshot support from equal wall-clock time.

To address these challenges, we propose a partial-mobility planar FAS that operationally separates acquisition from refinement. The fixed Nyquist seed supplies a persistent calibration reference and an ambiguity-controlled search region; the movable ports reconfigure the sparse aperture only after the seed data define a credible angular basin. Following D-optimal experimental design~\cite{kiefer_doptimal,fedorov_opt}, we derive the Fisher information matrix (FIM) of the resulting data-adaptive policy, a relaxed planar aperture upper bound, and a finite-port certificate that incorporates spacing, reachability, and actuator count. For multiple sources, an aperture-and-conditioning surrogate balances coordinate spread, steering-vector correlation, local coarray support, spacing, and movement cost; the exact two-state FIM then ranks and locally refines the candidates. Seed multiple signal classification (MUSIC) \cite{schmidt_music} is used for basin selection rather than final accuracy, after which a joint concentrated maximum-likelihood (ML) estimator uses the observations from both acquisition states. All baselines have the same total numbers of ports and RF chains and the same total snapshot budget. Movement and settling remain explicit feasibility and latency costs, so snapshot matching is not interpreted as zero-cost reconfiguration.

The main contributions are summarized as follows.

\begin{itemize}
    \item \textbf{Reconfigurable sparse-aperture architecture:} We transform a one-time planar sparse-array placement into a two-state sensing protocol with a fixed Nyquist acquisition seed, partially movable refinement ports, shared RF resources, explicit settling time, and a fixed total snapshot budget.
    \item \textbf{Adaptive information and finite-port retention:} We prove the policy-level FIM identity for observation-dependent movement, derive the relaxed rectangular-aperture D-optimal upper bound, and establish computable finite-port retention certificates under spacing, reachability, and actuator constraints.
    \item \textbf{Movement-aware design and multistate estimation:} We combine exact-FIM candidate ranking with an aperture-and-conditioning surrogate, smooth 2-D coarray support, spacing, and movement cost. The local refinement has a conditional Karush--Kuhn--Tucker (KKT) guarantee and explicit complexity, while confidence-gated seed MUSIC and joint two-state concentrated ML connect sparse-aperture control to estimation.
    \item \textbf{Boundary-aware validation:} Equal-budget experiments compare compact uniform planar, frozen sparse, random-reconfigured, frozen continuous, and all-movable arrays using MUSIC and local ML. Area, movement, port-count, source-separation, prior-shift, and position-error tests identify both the gains and the operating limits of reconfiguration.
\end{itemize}

The remainder of this paper is organized as follows. Section~\ref{sec:model} presents the planar FAS signal and reconfiguration model, difference coarray, and performance metrics. Section~\ref{sec:limits} analyzes the FIM and the fundamental geometry tradeoffs. Section~\ref{sec:opt} develops the seeded movement-constrained FAS reconfiguration and estimation method. Section~\ref{sec:robust} discusses practical movement and robustness constraints. Section~\ref{sec:sim} reports the simulation results, and Section~\ref{sec:conc} concludes the paper.

\emph{Notation}: Bold lowercase and uppercase letters denote vectors and matrices. $(\cdot)^T$, $(\cdot)^H$, and $(\cdot)^{-1}$ represent transpose, Hermitian transpose, and inverse. $\mathbb{E}\{\cdot\}$ denotes expectation, $\Re\{\cdot\}$ denotes the real part, and $\odot$ the Khatri--Rao product when applied column-wise. $\mathbf{I}_N$ is the $N\times N$ identity matrix.

\section{Planar FAS Signal and Reconfiguration Model}\label{sec:model}

\subsection{FAS Reconfiguration State}

\begin{figure}[h]
\centering
\includegraphics[width=0.7\columnwidth]{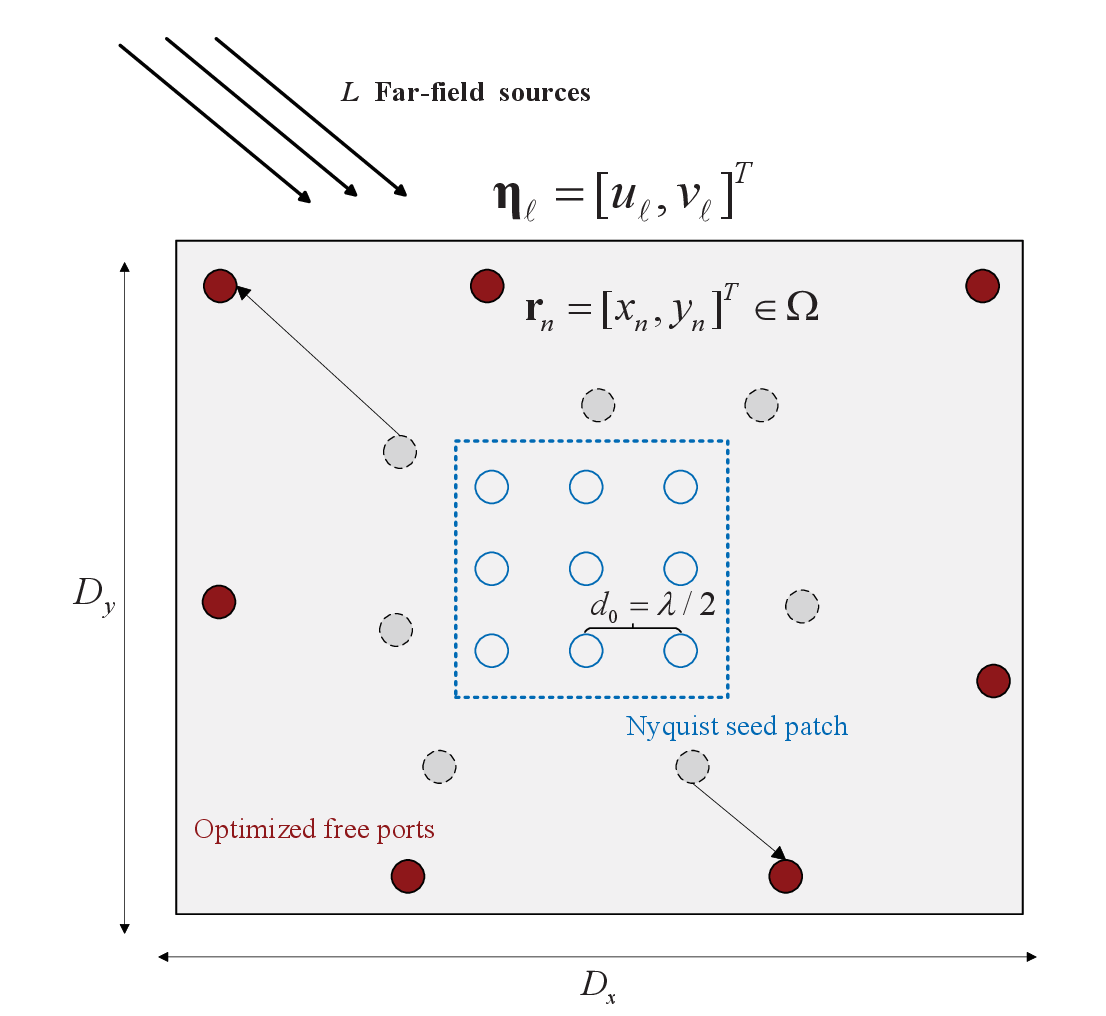}
	\caption{Seed-to-refinement planar FAS receiver, where the Nyquist seed remains fixed and the free ports move within the deployment region.}
	\label{fig:system_model}

\end{figure}

Figure~\ref{fig:system_model} illustrates the sensing model considered in this work. The receiver is a planar FAS terminal with $N_s$ calibrated seed ports and $N_f=N-N_s$ free fluid ports. The seed ports remain in a compact half-wavelength patch, while the free ports can move or switch among calibrated positions inside the same planar aperture. This structure can be implemented by liquid-metal or mechanically translated radiating ports, or by a pixel/switch-based FAS surface in which the active feed point is selected from a calibrated set. In the seed state, only the $N_s$ seed ports are connected; the free ports are parked or radiatively disabled and the corresponding RF branches remain idle. After reconfiguration and settling, the RF switch connects all $N$ ports to the refinement receiver. The main model uses $N_{\rm RF}=N$ simultaneous receive chains to isolate the geometry effect. Fewer chains can time-multiplex ports only when the probing waveform is repeated or phase coherent over the switching cycle; asynchronous passive sources instead require simultaneous chains or a calibrated hybrid combiner. The FAS gain therefore comes from changing calibrated locations under a fixed simultaneous front-end budget, rather than from adding antenna elements or RF chains.

A sensing frame is divided into a seed acquisition state and a refinement state. Let
\begin{equation}
    \mathbf{r}_n^{(m)}=[x_n^{(m)},y_n^{(m)}]^T\in\Omega,\quad
    \Omega=[0,D_x]\times[0,D_y],
\end{equation}
be the position of the $n$th port in state $m\in\{0,1\}$, where $m=0$ is the seed/standby state and $m=1$ is the post-movement refinement state. The seed ports satisfy
\begin{equation}
    \mathbf{r}_n^{(1)}=\mathbf{r}_n^{(0)}\in\mathcal{S},\quad n=1,\ldots,N_s,
\end{equation}
where $\mathcal{S}$ is a compact Nyquist patch. The free fluid ports obey the movement model
\begin{align}
    \mathbf{r}_n^{(1)}&\in\Omega,\quad n=N_s+1,\ldots,N,\\
    \|\mathbf{r}_n^{(1)}-\mathbf{r}_n^{(0)}\|_2&\leq d_{\rm mv},
    \quad n=N_s+1,\ldots,N,
\end{align}
or, equivalently, a total movement budget $\sum_{n=N_s+1}^{N}\|\mathbf{r}_n^{(1)}-\mathbf{r}_n^{(0)}\|_2\leq B_{\rm mv}$. For a liquid-metal implementation, this constraint represents actuation distance within the sensing frame; for a pixel or switch-based FAS, it represents switching among calibrated positions with an equivalent reconfiguration cost. The subsequent formulas use $\mathbf{r}_n$ for a generic calibrated FAS state, and the optimized geometry refers to the refinement state after movement.

\subsection{Spatial-Frequency Observation Model}

For $L$ narrowband far-field sources, define the direction cosines
\begin{equation}
    \bm{\eta}_\ell=[u_\ell,v_\ell]^T,\quad
    u_\ell=\sin\vartheta_\ell\cos\varphi_\ell,\quad
    v_\ell=\sin\vartheta_\ell\sin\varphi_\ell,
\end{equation}
where $\varphi_\ell$ is the azimuth angle and $\vartheta_\ell$ is the polar angle measured from the array broadside. At calibrated coordinates, the steering vector is
\begin{equation}\label{eq:steering2d}
    [\mathbf{a}(\bm{\eta}_\ell)]_n
    =
    \exp\!\left(j\frac{2\pi}{\lambda}\mathbf{r}_n^T\bm{\eta}_\ell\right),
    \quad n=1,\ldots,N.
\end{equation}

The active sets are $\mathcal{I}_0=\{1,\ldots,N_s\}$ and $\mathcal{I}_1=\{1,\ldots,N\}$. A snapshot from state $m$ follows
\begin{equation}\label{eq:signal2d}
    \mathbf{x}^{(m)}(t)=\mathbf{A}^{(m)}(\bm{\eta})\mathbf{s}^{(m)}(t)+\mathbf{n}^{(m)}(t),
\end{equation}
where $N_m=|\mathcal I_m|$, $\mathbf{A}^{(m)}$ contains the state-dependent steering vectors, and $\mathbf n^{(m)}(t)\sim\mathcal{CN}(\mathbf0,\sigma^2\mathbf I_{N_m})$. The source samples are unknown deterministic nuisance parameters. For $T_m$ snapshots, let
$\mathbf X^{(m)}=[\mathbf x^{(m)}(1),\ldots,
\mathbf x^{(m)}(T_m)]$,
$\mathbf S^{(m)}=[\mathbf s^{(m)}(1),\ldots,
\mathbf s^{(m)}(T_m)]$, and
$\mathbf P^{(m)}=\mathbf S^{(m)}(\mathbf S^{(m)})^H/T_m$.
The total snapshot budget over the seed and refinement states is
$T_{\rm tot}=T_0+T_1$. 
The simulations use the same diagonal $\mathbf P=\diag(P_1,\ldots,P_L)$ in both states. After settling, the coordinates are known and the covariance used for coarray interpretation is
\begin{equation}\label{eq:cov2d}
    \mathbf{C}_x^{(m)}(\bm{\eta})
    =
    \mathbf{A}^{(m)}(\bm{\eta})\mathbf{P}^{(m)}(\mathbf{A}^{(m)})^H(\bm{\eta})
    +\sigma^2\mathbf{I}_{N_m}.
\end{equation}

Let $\mathbf R^{(m)}
=[\mathbf r_1^{(m)},\ldots,\mathbf r_N^{(m)}]^T$
denote the port-coordinate matrix in state $m$.
For two layouts fixed before data acquisition, the independent snapshot blocks have additive equivalent spatial-frequency information after the state-specific waveform nuisance parameters are eliminated:
\begin{equation}\label{eq:multistate_fim}
    \mathbf{J}_{uv}^{\rm frame}
    =
    \sum_{m\in\mathcal{M}}
    \mathbf{J}_{uv}\!\left(\bm{\eta};\mathbf{R}^{(m)},\mathbf{P}^{(m)},T_m\right),
\end{equation}
where \(\mathcal M=\{0,1\}\) denotes the set of acquisition states.  
The seed term stabilizes basin selection, whereas the refinement term supplies most of the high-resolution information. A frozen array is obtained when the geometry does not depend on the seed data.

The actual receiver is adaptive since
$\mathbf R^{(1)}=\pi(\mathbf X^{(0)})$, where $\pi$
denotes the reconfiguration policy that maps the seed
observation to a feasible refinement layout. The next result distinguishes this joint experiment from simply evaluating a deterministic geometry at the true DOAs.

\begin{proposition}[FIM of a data-adaptive reconfiguration policy]\label{prop:adaptive_fim}
Let $\pi$ be a deterministic or externally randomized policy with no explicit dependence on the unknown DOAs. Assume regular likelihoods with parameter-independent support and state-specific waveform nuisance parameters. Then the equivalent FIM of the joint observation $(\mathbf X^{(0)},\mathbf X^{(1)})$ is
\begin{equation}\label{eq:adaptive_fim}
    \overline{\mathbf{J}}_{uv}^{\rm policy}
    =\mathbf{J}_{uv}^{(0)}+
    \mathbb{E}_{\mathbf{X}^{(0)}\mid\bm{\eta}}
    \!\left[\mathbf{J}_{uv}^{(1)}\!\left(\bm{\eta};\pi(\mathbf{X}^{(0)}),\mathbf{P}^{(1)},T_1\right)\right].
\end{equation}
The selected-layout sum is therefore a conditional design diagnostic, whereas \eqref{eq:adaptive_fim} is the pre-acquisition policy information.
\end{proposition}

\begin{proof}
The joint likelihood factors as
$p_{\bm\eta}(\mathbf X^{(0)})
p_{\bm\eta}(\mathbf X^{(1)}\mid\pi(\mathbf X^{(0)}))$.
Let
$\mathbf s_0=\nabla_{\bm\eta}\log p_{\bm\eta}(\mathbf X^{(0)})$
and
$\mathbf s_1=\nabla_{\bm\eta}\log
p_{\bm\eta}(\mathbf X^{(1)}\mid\pi(\mathbf X^{(0)}))$
denote the corresponding score vectors. The joint score is
$\mathbf s_0+\mathbf s_1$, and regularity gives
$\mathbb E[\mathbf s_1\mid\mathbf X^{(0)}]=\mathbf0$.
Both cross terms therefore vanish; the remaining score covariances give \eqref{eq:adaptive_fim}. State-wise elimination of the independent waveform nuisance blocks preserves the sum.
\end{proof}

For a realized seed observation, the policy selects a deterministic refinement layout, and substituting this layout into~(9) defines the selected-layout information diagnostic.
Under the assumptions of Proposition~1, averaging these matrices over $\mathbf X^{(0)}$ recovers the policy FIM in~(10). The former is used for conditional layout ranking and realized-frame evaluation, whereas the latter characterizes the adaptive reconfiguration rule before the seed data are observed.

Hence, the true-direction placement is used only as an oracle benchmark. An implementable controller forms its design distribution from a seed confidence region, posterior, or slow-time prior, and policy-level bounds average over the resulting seed uncertainty.

Conversion to physical angles uses the Jacobian
\begin{equation}
    \mathbf{B}_\ell
    =
    \frac{\partial (u_\ell,v_\ell)}{\partial(\vartheta_\ell,\varphi_\ell)}
    =
    \begin{bmatrix}
        \cos\vartheta_\ell\cos\varphi_\ell & -\sin\vartheta_\ell\sin\varphi_\ell\\
        \cos\vartheta_\ell\sin\varphi_\ell & \sin\vartheta_\ell\cos\varphi_\ell
    \end{bmatrix},
\end{equation}
with $\mathbf{B}=\diag(\mathbf{B}_1,\ldots,\mathbf{B}_L)$, the FIM for $(\vartheta,\varphi)$ is $\mathbf{B}^T\mathbf{J}_{uv}\mathbf{B}$.

\subsection{2-D Difference Coarray}

Suppressing the state index, covariance vectorization gives
\begin{equation}\label{eq:vec2d}
    \vecop(\mathbf{C}_x-\sigma^2\mathbf{I})
    =
    (\mathbf{A}^*\odot\mathbf{A})\mathbf{p},
\end{equation}
where $\mathbf p=[P_1,\ldots,P_L]^T$. Each virtual-array entry is indexed by a physical difference:
\begin{equation}
    \exp\!\left(j\frac{2\pi}{\lambda}(\mathbf{r}_i-\mathbf{r}_j)^T\bm{\eta}_\ell\right),
\end{equation}
so the 2-D difference coarray is
\begin{equation}\label{eq:coarray2d}
    \mathbb{D}_2(\mathbf{R})
    =
    \{\mathbf{r}_i-\mathbf{r}_j:1\leq i,j\leq N\}.
\end{equation}
If $\mathbb{D}_2$ contains a rectangular Nyquist patch
\begin{equation}
    \mathcal{G}(M_x,M_y)
    =
    \{(m_xd_0,m_yd_0): |m_x|\leq M_x,\ |m_y|\leq M_y\},
\end{equation}
with $d_0=\lambda/2$, the covariance contains a virtual uniform planar array (UPA) patch suitable for ambiguity-controlled subspace processing. Continuous locations rarely produce exact lattice lags, so Section~\ref{sec:opt} rewards differences near the desired patch rather than imposing exact equality.

\subsection{Frozen Arrays and FAS Reconfiguration Policies}

A frozen array selects one coordinate matrix $\mathbf R=[\mathbf r_1,\ldots,\mathbf r_N]^T$; a FAS policy selects a reachable refinement matrix from the current state. With $d_{ij}=\|\mathbf r_i-\mathbf r_j\|_2$, the grid and partial-mobility feasible sets are
\begin{align}
    \mathcal{R}_{\rm grid}
    &=
    \left\{\mathbf{R}\in\Omega^N:
    \mathbf{r}_n/d_0\in\mathbb{Z}^2,\ 
    d_{ij}\geq d_{\min}\right\},\\
    \mathcal{R}_{\rm FAS}^{\rm part}
    &=
    \left\{\mathbf{R}^{(1)}\in\Omega^N:
    \mathbf{r}_n^{(1)}=\mathbf{r}_n^{(0)},\ n\leq N_s,
    \right.\nonumber\\[-1mm]
    &\quad\left.
    d_{ij}^{(1)}\geq d_{\min},\ 
    C_{\rm mv}\leq B_{\rm mv}\right\}.
\end{align}
Here,
\begin{equation}
    C_{\rm mv}(\mathbf{R}^{(1)};\mathbf{R}^{(0)})
    =
    \sum_{n=N_s+1}^{N}
    \|\mathbf{r}_n^{(1)}-\mathbf{r}_n^{(0)}\|_2.
\end{equation}
Thus the FAS feasible set is state dependent. Setting $B_{\rm mv}=0$, or selecting $\mathbf R^{(1)}$ independently of seed data, recovers a frozen design. The distinction is operational: the FAS reconfigures only after coarse directional information becomes available.

The numerical baselines therefore include compact UPA, frozen sparse-grid, random-reconfigured, frozen continuous, and all-movable designs. Their estimators separate the benefit of geometry from the benefit of joint multistate processing.

\subsection{Performance Metrics}

All accuracy curves are reported in azimuth/polar angle. With
$\mathbf{B}=\diag(\mathbf{B}_1,\ldots,\mathbf{B}_L)$ denoting the Jacobian from angular parameters to spatial frequencies, the angular FIM is
$\mathbf{J}_{\vartheta\varphi}=\mathbf{B}^{T}\mathbf{J}_{uv}\mathbf{B}$, and the plotted Cramér-Rao bound (CRB) metric is
\begin{equation}
    \mathrm{CRB}_{2\mathrm{D}}(\mathbf{R})
    =
    \sqrt{
    \frac{1}{L}
    \tr\!\left(\mathbf{J}_{\vartheta\varphi}^{\dagger}(\bm{\eta};\mathbf{R})\right)
    }\frac{180}{\pi},
\end{equation}
where $(\cdot)^\dagger$ is the Moore--Penrose inverse. For Monte Carlo trials, we map each estimate to
$\hat{\bm{\alpha}}_\ell=[\hat{\vartheta}_\ell,\hat{\varphi}_\ell]^T$, and then the RMSE is computed after permutation matching:
\begin{equation}
    \mathrm{RMSE}_{2\mathrm{D}}
    =
    \frac{180}{\pi}
    \sqrt{
    \mathbb{E}\!\left[
    \min_{\pi}
    \frac{1}{L}\sum_{\ell=1}^L
    \|\hat{\bm{\alpha}}_{\pi(\ell)}-\bm{\alpha}_\ell\|_2^2
    \right]}.
\end{equation}
Azimuth errors are wrapped before assignment. The placement is optimized in direction-cosine space, but both CRB and RMSE are evaluated in the same physical-angle domain. Two-state CRBs use~\eqref{eq:multistate_fim} with the same total snapshot budget as the corresponding RMSE experiment.

\section{FIM and Fundamental Geometry Insights}\label{sec:limits}

\subsection{Conditional FIM Consistent With the ML Estimator}

The performance bound must use the same statistical model as the concentrated ML estimator. We therefore treat the source waveforms as unknown deterministic nuisance parameters and eliminate them by orthogonal projection~\cite{van_trees,kay}. For clarity, the state superscript is suppressed in this subsection. Let $\bm{\Pi}_{\mathbf{A}}^{\perp}=\mathbf{I}_N-\mathbf{A}(\mathbf{A}^H\mathbf{A})^{\dagger}\mathbf{A}^H$, $\mathbf{d}_{\ell,p}=\partial\mathbf{a}_\ell/\partial \eta_{\ell,p}$ for $p\in\{u,v\}$, and $\mathbf{P}=\mathbf{S}\mathbf{S}^H/T$ as the time-averaged empirical waveform covariance. For indices $i=(\ell,p)$ and $j=(m,q)$, the conditional FIM averaged over this empirical waveform covariance is
\begin{equation}\label{eq:fim_conditional}
    [\mathbf{J}_{uv}]_{ij}
    =\frac{2T}{\sigma^2}\Re\!\left\{
    \mathbf{d}_{\ell,p}^H\bm{\Pi}_{\mathbf{A}}^{\perp}
    \mathbf{d}_{m,q}\,[\mathbf{P}]_{m,\ell}
    \right\}.
\end{equation}

For the uncorrelated equal-power simulations below, we normalize $\mathbf{P}=\mathbf{I}_L$ and vary the signal-to-noise ratio (SNR) through $\sigma^2$. The steering derivatives are explicit. Let $k=2\pi/\lambda$ and $\mathbf{a}_\ell=\mathbf{a}(\bm{\eta}_\ell)$. Then
\begin{align}
    \frac{\partial\mathbf{a}_\ell}{\partial u_\ell}
    &= jk\,\diag(x_1,\ldots,x_N)\mathbf{a}_\ell,\\
    \frac{\partial\mathbf{a}_\ell}{\partial v_\ell}
    &= jk\,\diag(y_1,\ldots,y_N)\mathbf{a}_\ell.
\end{align}

The D-optimal objective is $\log\det\mathbf{J}_{uv}$, following the optimal experimental design criterion that maximizes the information-volume determinant~\cite{kiefer_doptimal,fedorov_opt}.

\subsection{Single-Source Closed Form}

The single-source case gives a useful design law. It isolates the pure aperture effect before source coupling, coarray coverage, and estimator initialization are introduced.

\begin{theorem}[Relaxed single-source planar D-optimal bound]\label{thm:single}
For one source with unknown spatial frequency $(u,v)$ under the conditional model in~\eqref{eq:fim_conditional}, the FIM for $(u,v)$ is proportional to the coordinate covariance matrix
\begin{equation}
    \bm{\Sigma}_{r}
    =
    \frac{1}{N}\sum_{n=1}^N
    (\mathbf{r}_n-\bar{\mathbf{r}})
    (\mathbf{r}_n-\bar{\mathbf{r}})^T.
\end{equation}
Therefore, the relaxed approximate-design problem over probability measures supported on $\Omega=[0,D_x]\times[0,D_y]$ satisfies
\begin{equation}\label{eq:single_det_bound}
    \det\bm{\Sigma}_r\leq \frac{D_x^2D_y^2}{16},
\end{equation}
with equality achieved by equal probability mass on the four rectangle corners. For distinct physical ports subject to minimum spacing, the bound is generally not attained; it is approached by balanced clusters near the four corners.
\end{theorem}

\begin{proof}
For one source, after removing the unknown complex amplitude, the spatial-frequency derivatives are $jkx_n a_n$ and $jky_n a_n$. Projection onto the orthogonal complement of the steering vector removes the common phase term, leaving only centered coordinates. Hence, the geometry-dependent FIM is $cN\bm{\Sigma}_r$ for a positive scalar $c$ determined by SNR, snapshots, and source power. Popoviciu's inequality gives $\operatorname{var}(x)\leq D_x^2/4$ and $\operatorname{var}(y)\leq D_y^2/4$, while Hadamard's inequality gives $\det\bm{\Sigma}_r\leq\operatorname{var}(x)\operatorname{var}(y)$. Equal corner mass has mean $(D_x/2,D_y/2)$, covariance $\diag(D_x^2/4,D_y^2/4)$, and zero cross-covariance, and thus, it attains the relaxed bound.
\end{proof}

The relaxed bound is an information upper bound, but a physical FAS has finitely many distinct ports and additional spacing and reachability constraints. The following certificate applies directly to that feasible set.

\begin{theorem}[Finite-port partial-mobility certificate]\label{thm:finite_retention}
Let $\mathcal F_M$ be the nonempty compact set of refinement layouts reachable when at most $M$ of the $N$ ports may move, including the box, movement-budget, and minimum-spacing constraints. Define
\begin{align}
 V_M&=\max_{\mathbf R\in\mathcal F_M}\det\bm\Sigma_r(\mathbf R),\\
 W_{p,M}&=\max_{\mathbf R\in\mathcal F_M,n}p_n-
           \min_{\mathbf R\in\mathcal F_M,n}p_n,\quad p\in\{x,y\},\\
 U_M&=W_{x,M}^2W_{y,M}^2/16.
\end{align}
For any feasible computed layout $\widehat{\mathbf R}_M\in\mathcal F_M$, let $L_M=\det\bm\Sigma_r(\widehat{\mathbf R}_M)$. Then
\begin{equation}\label{eq:finite_sandwich}
0\leq L_M\leq V_M\leq U_M.
\end{equation}
If the actuator sets are nested, $\mathcal F_M\subseteq\mathcal F_{M+1}$, then $V_M$ is nondecreasing. Moreover, for $M<N$, $L_N>0$, and $U_N>0$, the information retained relative to the finite-port all-movable optimum obeys
\begin{equation}\label{eq:finite_ratio_certificate}
 \frac{L_M}{U_N}
 \leq \frac{V_M}{V_N}
 \leq \min\!\left\{1,\frac{U_M}{L_N}\right\}.
\end{equation}
For a fixed seed and free-port subset, their exact finite-sample covariance decomposition is
\begin{equation}\label{eq:covariance_decomposition}
\bm\Sigma_r=\alpha\bm\Sigma_s+(1-\alpha)\bm\Sigma_f+
\alpha(1-\alpha)(\bm\mu_s-\bm\mu_f)(\bm\mu_s-\bm\mu_f)^T,
\end{equation}
where $\alpha=N_s/N$. Thus, an off-center reachable free-port cloud contributes an explicit between-group term rather than invalidating the certificate.
\end{theorem}

\begin{proof}
Every coordinate appearing in $\mathcal F_M$ lies in an interval of width $W_{p,M}$. Popoviciu's inequality and Hadamard's inequality therefore give $\det\bm\Sigma_r\leq U_M$. Feasibility of $\widehat{\mathbf R}_M$ gives $L_M\leq V_M$, proving \eqref{eq:finite_sandwich}. Set inclusion gives $V_M\leq V_{M+1}$. Combining $L_M\leq V_M$, $V_N\leq U_N$, $V_M\leq U_M$, and $L_N\leq V_N$ yields \eqref{eq:finite_ratio_certificate}. Equation~\eqref{eq:covariance_decomposition} is the standard within-group plus between-group covariance identity applied to the finite seed and free-port sets.
\end{proof}

Reachability supplies $U_M$, while any spacing-feasible returned layout supplies $L_M$; hence, physical constraints enter the certificate without assuming attainable corner collocation.

\begin{corollary}[Centered-seed relaxed retention benchmark]\label{cor:retention}
Let $\alpha=N_s/N$. Suppose the fixed seed and aperture share a center and $\bm\Sigma_s=\diag(\sigma_{s,x}^2,\sigma_{s,y}^2)$. Distributing the remaining relaxed mass equally over the corners produces
\begin{equation}
    \bm{\Sigma}_{\rm partial}
    =\alpha\bm{\Sigma}_s+(1-\alpha)
    \diag(D_x^2/4,D_y^2/4),
\end{equation}
and therefore
\begin{align}
\frac{\det\bm{\Sigma}_{\rm partial}}{D_x^2D_y^2/16}
={}&\left(1-\alpha+\frac{4\alpha\sigma_{s,x}^2}{D_x^2}\right)
\left(1-\alpha+\frac{4\alpha\sigma_{s,y}^2}{D_y^2}\right).
\label{eq:partial_retention}
\end{align}
As $D_x,D_y\rightarrow\infty$ with a fixed compact seed, the normalized determinant tends to $(1-\alpha)^2=(N_f/N)^2$.
\end{corollary}

\begin{proof}
The seed and balanced corner design have the same centroid, so no between-group covariance term appears. Their mixture covariance is the weighted sum in the statement. Taking its determinant gives~\eqref{eq:partial_retention}; normalized seed variances vanish as the aperture expands.
\end{proof}

Corollary~\ref{cor:retention} is an analytic aperture-scaling benchmark, not a finite-port optimum; Theorem~\ref{thm:finite_retention} supplies the physically feasible certificate.
The factor $(N_f/N)^2$ denotes the asymptotic determinant retention of the relaxed model, rather than a direct CRB ratio or finite-port guarantee.

\begin{remark}
Theorem~\ref{thm:single} explains the tendency toward corners. Finite spacing, multi-source conditioning, movement cost, and robustness prevent all physical ports from concentrating at the boundary.
\end{remark}

A planar array still has at most $N(N-1)+1$ distinct ordered differences. Reconfiguration does not remove this counting limit; it uses the available physical area to improve the FIM while retaining a small local lag patch for basin control. Because the partial-mobility feasible set is contained in the all-movable set, the latter is an information upper bound. The relevant hardware metric is information retained per actuated port.

\subsection{Multi-Source Conditioning}

With multiple sources, aperture alone is insufficient: nearby steering and derivative subspaces can become nearly collinear. Define
\begin{equation}
    \mathbf{G}(\bm{\eta};\mathbf{R})
    =
    \frac{1}{N}\mathbf{A}^H(\bm{\eta})\mathbf{A}(\bm{\eta}) ,
\end{equation}
as the normalized steering Gram matrix. A small $\lambda_{\min}(\mathbf G)$ reduces identifiability even when the coordinate covariance is large.

\begin{proposition}[Aperture and separation tradeoff]\label{prop:conditioning}
For two equal-power sources separated by $\Delta\bm\eta$, define
\begin{equation}
    \gamma(\Delta\bm{\eta})
    =
    \frac{1}{N}\sum_{n=1}^N
    \exp\!\left(j\frac{2\pi}{\lambda}\mathbf{r}_n^T\Delta\bm{\eta}\right).
\end{equation}
The normalized two-source Gram matrix has eigenvalues $1\pm|\gamma(\Delta\bm\eta)|$. Aperture growth is therefore useful only when the layout also suppresses $|\gamma|$ over the relevant separations.
\end{proposition}

\begin{proof}
After column normalization, the Gram matrix equals
$\begin{bmatrix}1&\gamma\\ \gamma^*&1\end{bmatrix}$. Direct eigendecomposition gives $1\pm|\gamma|$; as $|\gamma|\rightarrow1$, the derivative subspaces and the conditional FIM become ill-conditioned.
\end{proof}

This tradeoff pulls some free ports from the corners toward edges or the interior, where they reduce average correlation over the design sector.

\subsection{Global-Basin Control for Large Apertures}

A large FAS aperture provides high Fisher information because the steering phase changes rapidly with $\bm{\eta}$. The same phase sensitivity narrows the main lobes and introduces additional sidelobe structure in MUSIC and ML objectives. If the largest port separation along a direction $\mathbf{q}$ is $D_{\mathbf{q}}$, the local peak spacing is on the order of $\lambda/D_{\mathbf{q}}$. The fixed seed prevents exact visible-sector aliases in the ideal manifold, but finite-snapshot perturbations and multi-source peak merging can still make unrestricted initialization fragile.

The seed patch changes this ambiguity structure. Its half-wavelength spacing makes the steering map essentially unambiguous over the intended visible sector, but its aperture is too small for high final accuracy. The proposed estimator uses the seed primarily for basin selection and retains its likelihood contribution when jointly refining both states. This mechanism combines stable initialization with large-aperture precision.

\begin{proposition}[Wide-sector search and local high-resolution refinement]\label{prop:wide_search}
Consider a rectangular seed patch containing at least two positions along each coordinate, with spacing $d_0=\lambda/2$ in both planar directions. For spatial frequencies in the open visible disk $\mathcal{V}=\{(u,v):u^2+v^2<1\}$, two distinct directions cannot produce identical adjacent phase progressions on the seed patch, even after the unknown source phase is absorbed into a complex amplitude. Hence, the seed state preserves the ambiguity-free scan region of a compact UPA. After the seed selects a local basin, the refinement state can use larger nonuniform baselines to reduce the local CRB without requiring a global large-aperture search.
\end{proposition}

\begin{proof}
Equality of the steering vectors up to an unknown complex scalar requires equality of their adjacent phase progressions. Thus, $k d_0 \Delta u=2\pi p$ and $k d_0 \Delta v=2\pi q$ for integers $p$ and $q$, where $\Delta u$ and $\Delta v$ are the spatial-frequency differences. Since $k d_0=\pi$, this implies $\Delta u=2p$ and $\Delta v=2q$. In the open visible region, $|\Delta u|<2$ and $|\Delta v|<2$, so the only feasible integers are $p=q=0$. The two directions are therefore identical. The CRB reduction after refinement then follows locally from the coordinate-covariance aperture law in Theorem~\ref{thm:single} and its multi-source extension through Proposition~\ref{prop:conditioning}.
\end{proof}


\section{Movement-Constrained FAS Reconfiguration and Estimation}\label{sec:opt}

\subsection{Seeded Reconfiguration Protocol}

A fully movable large-aperture FAS is attractive from an information perspective but poses initialization and hardware challenges. We therefore operate the FAS in two states. In the seed state, $N_s$ calibrated ports form a compact Nyquist patch and the free ports remain in known standby positions. In the refinement state, only the free fluid ports move or switch to optimized positions, while the seed ports remain fixed as a calibrated local reference. The seed set is
\begin{equation}
    \mathcal{S}
    =
    \{(x_0+m_xd_0,y_0+m_yd_0):m_x,m_y=0,\ldots,M_s-1\},
\end{equation}
where $N_s=M_s^2$. These seed ports are not a separate array. They are FAS positions intentionally kept in a local half-wavelength patch during sensing so that the receiver has an ambiguity-controlled reference manifold before the free ports move. The free ports start from standby coordinates $\mathbf{r}_n^{(0)}$ and move to refinement coordinates $\mathbf{r}_n^{(1)}$ inside $\Omega$.

This architecture assigns the two parts of the FAS different roles:
\begin{itemize}
    \item the seed state supplies a local, Nyquist-sampled 2-D manifold for coarse MUSIC and basin selection;
    \item the moved free fluid ports supply the large coordinate covariance needed to reduce the CRB during final refinement.
\end{itemize}
The seed mitigates ambiguity, whereas free-port reconfiguration improves local estimation precision.

This partial-reconfiguration structure should not be interpreted as a claim that fixing seed ports always gives a lower CRB than moving all ports. Let $\mathcal{R}_M$ be the set of refinement geometries that can be reached when at most $M$ ports are actuated after the seed observation. If $M_1<M_2$, then $\mathcal{R}_{M_1}\subseteq\mathcal{R}_{M_2}$, and the best achievable CRB under the larger feasible set cannot be worse. A fully movable FAS therefore defines the information upper bound. Fixed seed ports instead provide a half-wavelength wide-sector reference, reduce calibration and settling overhead, and retain a stable reference geometry for local ML. The design objective is a hardware-aware accuracy--cost tradeoff, not an attempt to outperform an unconstrained all-movable array.

\subsection{FIM-Informed D-Optimal Movement-Regularized Reconfiguration}

Let $\mathcal Q=\{\bm\eta^{(q)}\}_{q=1}^Q$ sample a seed-derived confidence region or slow-time prior. It cannot use unknown true directions except in an oracle benchmark. For refinement matrix $\mathbf R^{(1)}$, define the frame-level criterion
\begin{equation}
    \mathcal{L}_{\rm D}(\mathbf{R}^{(1)})
    =
    \frac{1}{Q}\sum_{q=1}^Q
    \log\det\!\left(
        \mathbf{J}_{uv}^{\rm frame}(\bm{\eta}^{(q)};\mathbf{R}^{(0)},\mathbf{R}^{(1)})+\epsilon\mathbf{I}
    \right).
\end{equation}
Although seed information is independent of $\mathbf R^{(1)}$, it enters the log-determinant jointly with refinement information and cannot be dropped. The controller solves
\begin{equation}\label{eq:main_opt}
\begin{aligned}
    \max_{\mathbf{R}^{(1)}} \quad
        & \mathcal{L}_{\rm D}(\mathbf{R}^{(1)})
        + \rho\Phi_{\rm ca}(\mathbf{R}^{(1)})
        -\mu\Psi_{\rm sp}(\mathbf{R}^{(1)})\\
        &-\chi\Psi_{\rm mv}(\mathbf{R}^{(1)};\mathbf{R}^{(0)})\\
    \text{s.t.}\quad
        & \mathbf{r}_n^{(1)}\in\Omega,\quad n=N_s+1,\ldots,N,\\
        & \{\mathbf{r}_1^{(1)},\ldots,\mathbf{r}_{N_s}^{(1)}\}=\mathcal{S},\\
        & C_{\rm mv}(\mathbf{R}^{(1)};\mathbf{R}^{(0)})
          \leq B_{\rm mv}.
\end{aligned}
\end{equation}
The first two terms reward frame information and local lag support, while the others penalize spacing violations and movement. The movement penalty is
\begin{equation}
    \Psi_{\rm mv}(\mathbf{R}^{(1)};\mathbf{R}^{(0)})
    =
    \sum_{n=N_s+1}^{N}
    \|\mathbf{r}_n^{(1)}-\mathbf{r}_n^{(0)}\|_2^2.
\end{equation}
With a large $B_{\rm mv}$ and $\chi=0$, the design is information driven; with $B_{\rm mv}=0$, it is frozen.

To avoid evaluating the exact FIM throughout every multi-start search, candidates are generated using
\begin{equation}\label{eq:fim_surrogate}
\widetilde{\mathcal{L}}_{\rm D}(\mathbf{R})
=\log\det(\bm{\Sigma}_r+\epsilon\mathbf{I}_2)
-\frac{\kappa}{Q}\sum_{q=1}^{Q}
\sum_{\ell<m}
\left|\frac{\mathbf{a}_\ell^H\mathbf{a}_m}{N}\right|^2,
\end{equation}
which combines the aperture law with a multi-source correlation penalty. Every candidate is then rescored by the exact multistate objective, and the best candidate is used to initialize local refinement of~\eqref{eq:main_opt}. Reported layouts and CRBs therefore use the matched FIM, not the surrogate.

The smooth 2-D coarray coverage metric is
\begin{align}
    \Phi_{\rm ca}(\mathbf{R})
    &=
    \frac{1}{|\mathcal{G}_0|}\sum_{\mathbf{m}\in\mathcal{G}_0}
    \log\!\left(S_{\mathbf{m}}+\epsilon_c\right),
    \label{eq:coarray_reg}\\
    S_{\mathbf{m}}
    &=\sum_{i,j}\exp\!\left[
    -\frac{\|\mathbf{r}_i-\mathbf{r}_j-d_0\mathbf{m}\|_2^2}{2\tau^2}
    \right],
\end{align}
where $\mathcal G_0$ is a small target Nyquist patch. The logarithmic soft count rewards coverage without repeatedly rewarding the same lag and remains differentiable for continuous positions. It serves as a geometric regularizer for local MUSIC candidate generation, while the reported CRBs and final concentrated-likelihood estimates use the physical array
manifold.

The spacing penalty is
\begin{equation}
    \Psi_{\rm sp}(\mathbf{R})
    =
    \sum_{i<j}\left[d_{\min}-\|\mathbf{r}_i-\mathbf{r}_j\|_2\right]_+^2.
\end{equation}
Box-only cases use multi-start limited-memory BFGS with box constraints (L-BFGS-B)~\cite{nocedal}; an active movement constraint is handled by sequential least-squares quadratic programming (SLSQP). Projected-gradient or sequential convex approximation can also optimize the smooth surrogate~\cite{boyd_convex}.

The implementation initializes corner, edge, and interior layouts, searches with the surrogate, ranks with the exact objective, and refines the best candidate. A returned layout must pass movement and spacing checks; repaired layouts are rescored.

\subsection{Gradient Structure and Feasibility Handling}

For any movable coordinate $z\in\{x_n^{(1)},y_n^{(1)}\}$, the log-determinant derivative is
\begin{equation}\label{eq:logdet_grad}
    \frac{\partial}{\partial z}
    \log\det(\mathbf{J}_{uv}^{\rm frame}+\epsilon\mathbf{I})
    =
    \tr\!\left[
    (\mathbf{J}_{uv}^{\rm frame}+\epsilon\mathbf{I})^{-1}
    \frac{\partial\mathbf{J}_{uv}^{\rm frame}}{\partial z}
    \right].
\end{equation}
Only the refinement-state FIM depends on $z$. The implementation uses finite differences for candidate generation and exact-objective refinement;~\eqref{eq:logdet_grad} also permits analytic derivatives.

The remaining terms are differentiable, while box and movement constraints are imposed by the solver. Spacing is penalized during search and checked afterward. Vanishing projected-gradient or KKT residuals certify only local first-order stationarity~\cite{nocedal}; the problem remains nonconvex.

\begin{proposition}[Stationarity of exact local refinement]\label{prop:stationarity}
Assume that the steering matrix has constant column rank in the local design region, $\epsilon>0$, accepted ports do not coincide, and the active constraints satisfy a standard constraint qualification. If the exact-objective refinement of~\eqref{eq:main_opt} uses a line-search sequential quadratic programming (SQP) or projected method with sufficient ascent and exact or asymptotically accurate finite-difference derivatives, every accumulation point is a first-order KKT point of the local constrained problem.
\end{proposition}

\begin{proof}
The stated assumptions make the objective continuously differentiable on a compact level set. Standard line-search convergence makes the first-order residual vanish along convergent subsequences, and constraint qualification gives the KKT system~\cite{nocedal}.
\end{proof}

This local result gives no global guarantee; the multi-start surrogate only provides candidate initializations.

\begin{remark}[Layout-design complexity]\label{rem:layout_complexity}
Let $K$ be the number of starts, $I_s$ and $I_e$ the surrogate and exact iterations, and $|\mathcal G_0|$ the number of target lags. Surrogate and exact-FIM evaluations cost $O(QNL^2+N^2|\mathcal G_0|)$ and $O(Q(NL^2+L^3))$, respectively. Finite-difference refinement therefore costs
\begin{align}
O\big(&KI_s[QNL^2+N^2|\mathcal G_0|]+KQ[NL^2+L^3]\notag\\
&+I_eN_fQ[NL^2+L^3]\big).
\end{align}
Analytic derivatives remove the $N_f$ multiplier; slowly varying layouts can be cached.
\end{remark}

\subsection{Seeded FAS-ML Estimator}

The seed observation supplies a coarse, ambiguity-controlled MUSIC estimate:
\begin{equation}
    \tilde{\bm{\eta}}_1,\ldots,\tilde{\bm{\eta}}_L
    =
    \arg\max_{\bm{\eta}\in\mathcal{U}}
    \frac{1}{\mathbf{a}_s^H(\bm{\eta})\mathbf{E}_{n,s}\mathbf{E}_{n,s}^H\mathbf{a}_s(\bm{\eta})},
\end{equation}
where $\mathbf{a}_s$ and $\mathbf{E}_{n,s}$ are the seed-patch steering vector and noise subspace.

Let $\gamma_s$ be the weakest selected peak divided by the median seed-spectrum level. Using the weakest peak makes the gate sensitive to the least reliable source, so source-conditioned movement is triggered only when all selected peaks are sufficiently prominent. Below a prescribed threshold, movement is skipped. 
Otherwise, the free ports move to $\mathbf R^{(1)\star}$ and acquire refinement data. The seed peaks define local neighborhoods, post-movement MUSIC proposes candidates within them, and the best candidate initializes concentrated ML. With $\hat{\mathbf C}_x^{(m)}=\mathbf X^{(m)}(\mathbf X^{(m)})^H/T_m$, the final estimate is
\begin{equation}\label{eq:ml}
    \hat{\bm{\eta}}
    =
    \arg\min_{\bm{\eta}\in\mathcal{B}(\tilde{\bm{\eta}},\Delta)}
    \sum_{m=0}^{1}T_m\tr\!\left[
        \left(\mathbf{I}_{N_m}-\mathbf{P}_{\mathbf{A}^{(m)}(\bm{\eta})}\right)
        \hat{\mathbf{C}}_x^{(m)}
    \right],
\end{equation}
where $\mathbf P_{\mathbf A^{(m)}}=\mathbf A^{(m)}
((\mathbf A^{(m)})^H\mathbf A^{(m)})^{\dagger}
(\mathbf A^{(m)})^H$ and
$\mathcal B(\tilde{\bm\eta},\Delta)$ is a local box around
the selected candidate, with $\Delta$ denoting the search
half-width in the direction-cosine domain. The state weights match the additive information model: the seed suppresses remote basins and the moved aperture supplies fine differential-phase information.

\begin{algorithm}[!t]
    \caption{Joint 2-D DOA Estimation after Partial-Mobility Reconfiguration}
    \label{alg:seeded_ml}
    \LinesNumbered
    \KwIn{Seed data $\mathbf{X}^{(0)}$, current layout $\mathbf{R}^{(0)}$, seed set $\mathcal{I}_s$, source count $L$}
    \KwOut{2-D DOA estimates $\hat{\bm{\eta}}_\ell=(\hat u_\ell,\hat v_\ell)$}
    Form the seed covariance and obtain $\tilde{\bm\eta}_1,\ldots,\tilde{\bm\eta}_L$ and confidence $\gamma_s$ by 2-D MUSIC\;
    \If{$\gamma_s$ is below the refinement threshold}{
        \Return the seed estimates $\tilde{\bm{\eta}}_\ell$\;
    }
    Retrieve or solve \eqref{eq:main_opt}, move free ports to $\mathbf R^{(1)\star}$, and acquire $\mathbf X^{(1)}$ after settling\;
    Generate refinement-state MUSIC candidates inside the seed-gated neighborhoods\;
    Select the candidate with the lowest concentrated-likelihood cost and minimize \eqref{eq:ml} over $\mathcal B(\tilde{\bm\eta},\Delta)$\;
    Map $\hat{\bm{\eta}}_\ell$ to azimuth/polar angle if required\;
    \Return $\hat{\bm{\eta}}_\ell$, $\ell=1,\ldots,L$\;
\end{algorithm}

\begin{remark}[Complexity]
For a $G$-point grid, seed MUSIC costs $O(N_s^3+N_s^2G)$. After covariance formation, each two-state likelihood evaluation costs $O(N^2L+NL^2+L^3)$. A small candidate set and local refinement replace exhaustive multi-source grid enumeration; slowly varying deployments can retrieve layouts from an offline sector codebook.
\end{remark}

\section{Robustness and Practical Constraints}\label{sec:robust}

The DOAs must remain approximately constant over the seed, reconfiguration, and refinement interval.
Let $T_{\rm s}$ be the duration of one snapshot
and let $T_{\rm mv}$ include the movement or switching and
subsequent settling time. The angular-coherence requirement is
then
\begin{equation}
      T_0T_{\rm s}+T_{\rm mv}+T_1T_{\rm s}
    \leq T_{\rm coh}^{\rm ang},
\end{equation}
where $T_{\rm coh}^{\rm ang}$ is the angular coherence time. The controller can use an offline sector codebook when online optimization or reconfiguration cannot meet the latency requirement.
Minimum spacing is handled by the penalty and acceptance check in Section~\ref{sec:opt}.

For implemented positions $\mathbf r_n+\Delta\mathbf r_n$, let $\Delta\phi_n=(2\pi/\lambda)\Delta\mathbf r_n^T\bm\eta$. If independent port errors are zero-mean Gaussian with covariance $\bm\Sigma_p$, their covariance-domain coherence satisfies
\begin{equation}
    \mathbb{E}\!\left[e^{j(\Delta\phi_i-\Delta\phi_j)}\right]
    =
    \exp\!\left(
        -\left(\frac{2\pi}{\lambda}\right)^2
        \bm{\eta}^T\bm{\Sigma}_p\bm{\eta}
    \right),\qquad i\ne j.
\end{equation}
Position error therefore causes direction-dependent coherence loss; robustness can be introduced by averaging the log-determinant in~\eqref{eq:main_opt} over position-error samples. Geometry-dependent coupling is handled similarly by replacing $\mathbf a(\bm\eta)$ with a calibrated $\mathbf C(\mathbf R)\mathbf a(\bm\eta)$ in the FIM and likelihood.

\section{Simulation Results}\label{sec:sim}

The numerical study addresses four questions: whether FAS movement converts available surface area into Fisher information, whether a practical estimator can exploit the moved aperture, how partial actuation trades accuracy for movement cost, and whether reconfiguration remains useful when the source prior changes. Unless otherwise stated, $\lambda=1$, $d_0=\lambda/2$, $N=16$, $N_s=9$, and the deployment region is a square with side length $D=24d_0$. The reported SNR is the per-source ratio $P_\ell/\sigma^2$. Two equal-power sources have spatial frequencies
\begin{equation}\label{eta}
    \bm{\eta}_1=[0.18,0.10]^T,\quad
    \bm{\eta}_2=[0.29,0.17]^T.
\end{equation}
The proposed FAS uses a $3\times3$ seed and seven free ports. In every estimator comparison, the total observation budget $T_{\rm tot}$ is fixed: the FAS assigns $T_0=T_1=T_{\rm tot}/2$ snapshots to the seed and refinement states, whereas each frozen array uses all $T_{\rm tot}$ snapshots at its single geometry. The two-state ML objective and frame CRB use both FAS states, preventing a snapshot-count advantage. This comparison isolates spatial reconfiguration; under a fixed wall-clock deadline, a frozen array could collect additional samples during $T_{\rm mv}$, and the resulting latency tradeoff depends on the implementation described in Section~\ref{sec:robust}. The nonconvex placement problem is solved by multi-start surrogate search, exact conditional-FIM ranking, and exact local refinement. Unless varied explicitly, the movement budget reaches the selected refinement state. Table~\ref{tab:sim_params} summarizes the default setting.

\begin{table}[!t]
\centering
\caption{Default Simulation Parameters}
\label{tab:sim_params}
\renewcommand{\arraystretch}{1.16}
\footnotesize
\begin{tabular}{@{}p{0.42\columnwidth}p{0.48\columnwidth}@{}}
\toprule
\textbf{Parameter} & \textbf{Value} \\
\midrule
Wavelength and Nyquist spacing & $\lambda=1$, $d_0=\lambda/2$ \\
Ports and seed patch & $N=16$, $N_s=9$ ($3\times3$) \\
  RF-chain model & $N_{\rm RF}=N$ simultaneous chains; coherent switching is an extension \\
Deployment region & $24d_0\times24d_0$ square \\
  Sources & $L=2$, equal power, uncorrelated; per-source SNR \\
  Snapshot allocation & default $T_{\rm tot}=300$; $T_0=T_1=T_{\rm tot}/2$ for FAS \\
  SNR range & specified in each experiment \\
  Monte Carlo averaging & $300$ trials per point; $500$ for the SNR sweep \\
  Placement parameters & $d_{\min}=0.2\lambda$, $\rho=0.04$, $\mu=2000$, $\chi=0.006$ \\
  Coarray/search parameters & $\mathcal{G}_0=\mathcal{G}(4,4)$, $\tau=0.08\lambda$, $61\times49$ MUSIC grid \\
Baselines & UPA MUSIC/local ML, frozen sparse-grid MUSIC/local ML, random-reconfigured FAS, best frozen continuous array, direct FAS MUSIC, CRB \\
\bottomrule
\end{tabular}
\end{table}

\subsection{Reconfiguration and Basin Formation}

Figure~\ref{fig:geometry} shows the FAS movement pattern from compact standby positions to the optimized refinement state, together with the rounded 2-D difference coarray after movement. The free ports move toward corners and edges, while the seed ports remain compact and calibrated. This is the behavior predicted by the analysis. Boundary movement enlarges coordinate covariance and strengthens the dominant FIM terms, whereas the seed state preserves a local Nyquist structure for coarse basin selection. The figure therefore illustrates the central FAS mechanism: the seed mitigates ambiguity, whereas free-port movement improves local estimation precision.

\begin{figure}[h]
    \centering
    \includegraphics[width=0.99\columnwidth]{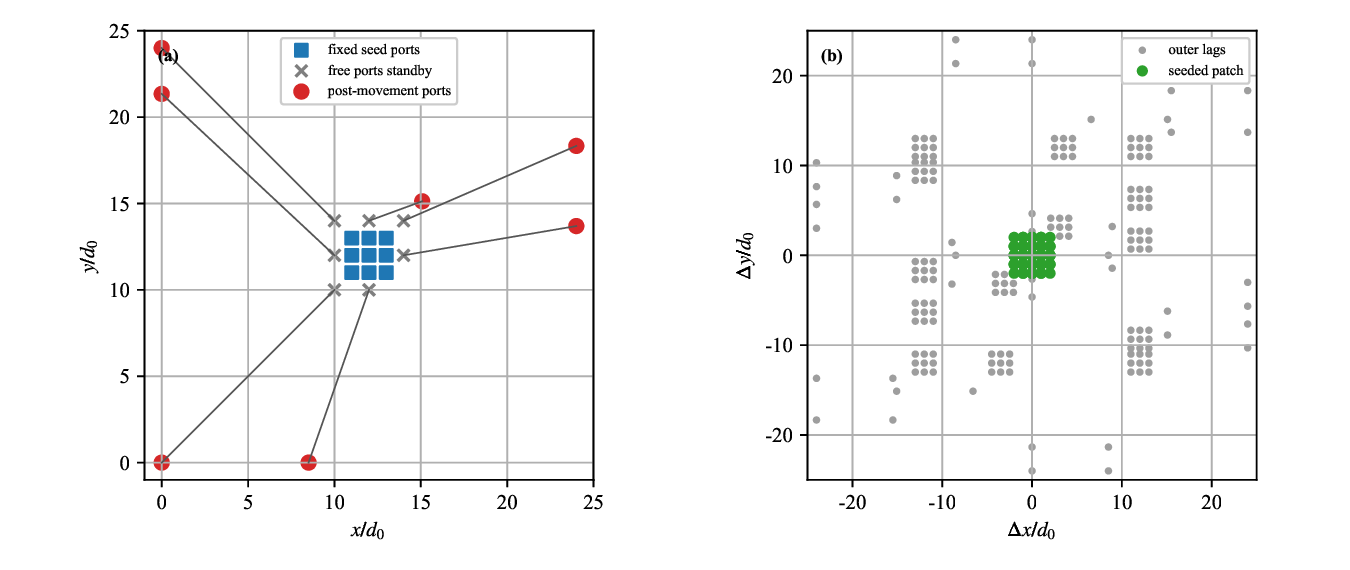}
    \caption{Seed-to-refinement movement and rounded 2-D difference coarray.}
    \label{fig:geometry}
\end{figure}

Figure~\ref{fig:spectrum_aliasing} separates global-basin information from local resolution. The post-movement FAS produces narrow peaks at the true directions together with remote sidelobe structure, whereas the seed spectrum is broader and forms a connected coarse basin. For the well-separated default pair, direct FAS MUSIC can exploit the narrow peaks; the seed gate becomes more important when the peaks merge or finite-snapshot perturbations make initialization unreliable. The two states are therefore complementary rather than individually sufficient over all source separations.

\begin{figure}[!t]
    \centering
    \includegraphics[width=1\columnwidth]{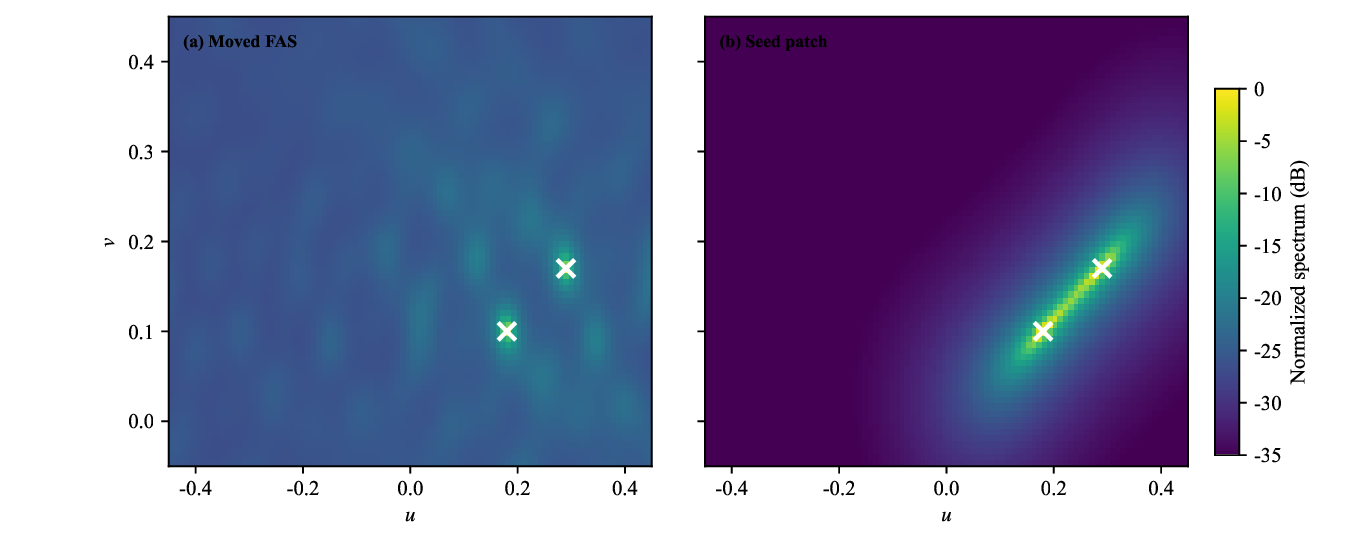}
    \caption{Post-movement and seed-state 2-D MUSIC spectra at SNR $=20$ dB.}
    \label{fig:spectrum_aliasing}
\end{figure}

\subsection{Aperture and Movement Scaling}

Figure~\ref{fig:crb_area} varies the deployment side length $D$ while keeping $N=16$. The compact UPA and the frozen sparse grid are nearly unchanged because their physical coordinates do not adapt to the newly available region. The reconfigured FAS CRB decreases as $D$ grows, confirming that port movement can convert physical area into Fisher information. This is an information-theoretic result, not a complete estimator result; its practical value depends on the seed-controlled refinement evaluated in the following figures.

\begin{figure}[h]
    \centering
    \includegraphics[width=0.98\columnwidth]{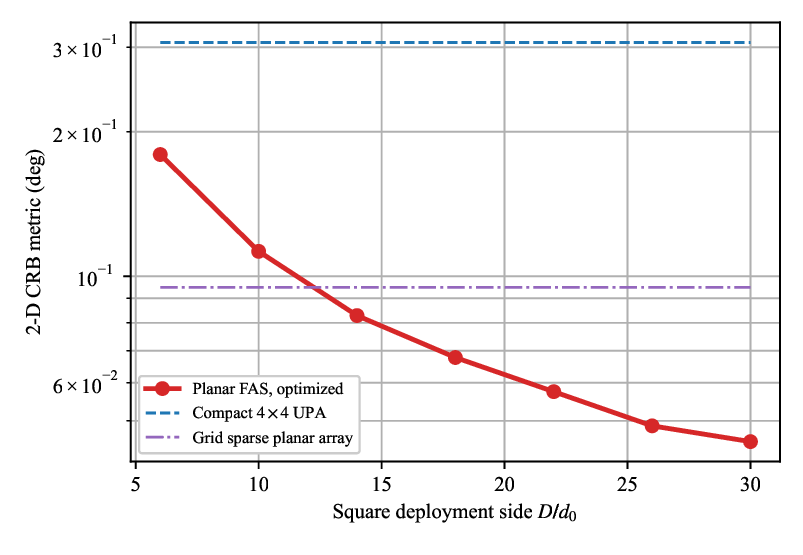}
    \caption{2-D CRB versus the square deployment side length.}
    \label{fig:crb_area}
\end{figure}

Figure~\ref{fig:movement_budget} isolates the effect of physical reconfiguration. At zero movement, the receiver remains a compact frozen aperture. Increasing $B_{\rm mv}$ creates longer nonredundant baselines, so all three grouped metrics decrease. The adjacent bars also show that the empirical root-mean-square error (RMSE) tracks the test-pair CRB, while the prior-averaged CRB measures performance beyond the selected pair. The progressive gain along the feasible movement path distinguishes reconfiguration from a favorable one-time placement.

\begin{figure}[h]
    \centering
    \includegraphics[width=0.95\columnwidth]{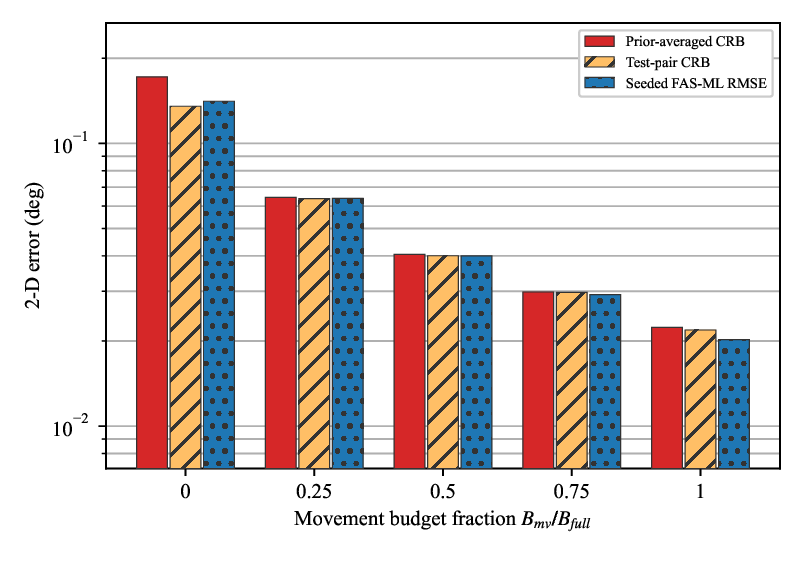}
    \caption{Prior-averaged CRB, test-pair CRB, and seeded FAS-ML RMSE versus the normalized movement budget $B_{\mathrm{mv}}/B_{\mathrm{full}}$.}
    \label{fig:movement_budget}
\end{figure}

\subsection{Estimation Accuracy and Resolution}

Figure~\ref{fig:accuracy_suite} summarizes the three estimator tests. Panel~(a) shows that increasing SNR cannot remove the aperture-dependent gap between the compact UPA, frozen sparse grid, and reconfigured FAS. Panel~(b) reaches the same conclusion under increasing sample support: all estimators improve, but the geometry gap remains under the equal total-snapshot budget. Panel~(c) normalizes every close-source RMSE by seeded FAS-ML. The first source is fixed, while the second source is moved toward it along the direction defined by the default pair in~\eqref{eta} to obtain the specified angular separations. Direct FAS MUSIC incurs a large threshold effect at the two smallest separations, whereas it approaches the proposed estimator after the peaks become resolvable.
This separates the two FAS mechanisms: movement reduces local estimation error, and the fixed seed protects basin selection.

\begin{figure*}[!t]
	\centering
	\subfloat[SNR scaling.\label{fig:accuracy_snr}]{%
		\includegraphics[width=0.318\textwidth]{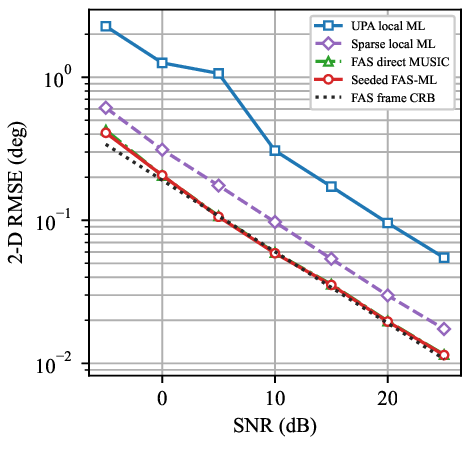}}%
	\hfill
	\subfloat[Snapshot scaling.\label{fig:accuracy_snapshots}]{%
		\includegraphics[width=0.318\textwidth]{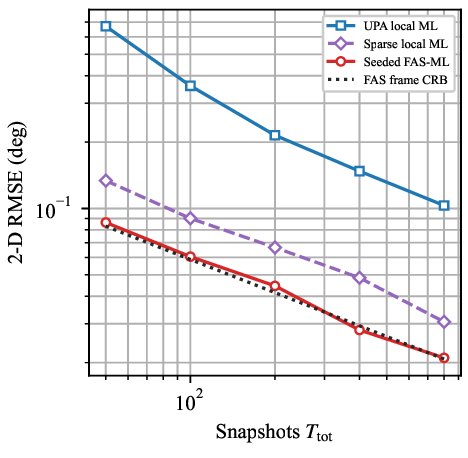}}%
	\hfill
	\subfloat[Close-source penalty.\label{fig:accuracy_resolution}]{%
		\includegraphics[width=0.318\textwidth]{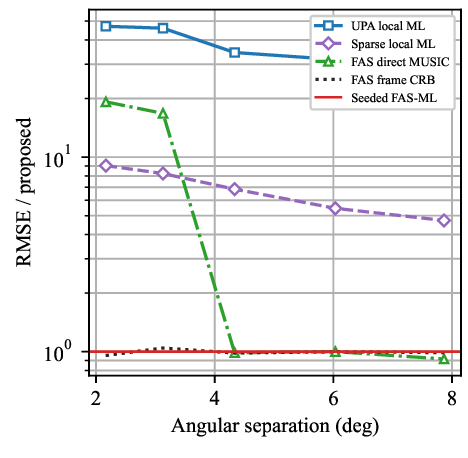}}
	\caption{Compact accuracy validation versus the SNR, snapshots, and source separation.}
	\label{fig:accuracy_suite}
\end{figure*}

Table~\ref{tab:rmse} reports a representative high-SNR operating point. The reconfigured FAS estimators have a clear margin over the strongest frozen-array local-ML baseline, while direct FAS MUSIC remains competitive for this resolved source pair. The table therefore confirms the aperture benefit while showing that the estimator advantage is not universal.

\begin{table}[!t]
\centering
\caption{Equal-Budget 2-D RMSE at SNR $=15$ dB}
\label{tab:rmse}
\renewcommand{\arraystretch}{1.18}
\begin{tabular}{@{}lcc@{}}
\toprule
\textbf{Method} & \textbf{RMSE (deg)} & \textbf{Relative to proposed} \\
\midrule
Compact $4\times4$ UPA MUSIC & $0.1800$ & $5.10\times$ \\
Compact $4\times4$ UPA local ML & $0.1719$ & $4.87\times$ \\
Frozen sparse-grid MUSIC & $0.0540$ & $1.53\times$ \\
Frozen sparse-grid local ML & $0.0538$ & $1.52\times$ \\
Direct FAS MUSIC & $0.0357$ & $1.01\times$ \\
\textbf{Seeded FAS-ML} & $\mathbf{0.0353}$ & $\mathbf{1.00\times}$ \\
FAS frame CRB & $0.0339$ & --- \\
\bottomrule
\end{tabular}
\end{table}

\subsection{Partial Actuation and Design Ablations}

Figure~\ref{fig:n_ports} varies the number of movable free ports $N_f=N-N_s$ while the compact $3\times3$ seed patch is kept fixed for both random and optimized FAS baselines.
This comparison is deliberately seed-matched: the random baseline has the same ambiguity-control seed as the proposed design and randomizes only the additional free-port locations. The compact UPA improves slowly because added ports mainly densify a local grid. The seeded random FAS benefits from aperture spreading but still uses the available baselines inefficiently since its post-movement state is not matched to the source sector. The seeded optimized FAS consistently gives the lowest CRB, showing that the gain comes from information-driven free-port reconfiguration rather than random mobility alone. The curve also provides a hardware-design insight: the first few movable free ports are especially valuable because they add long nonredundant baselines after the coarse state is known, while later ports mainly improve conditioning and robustness.

\begin{figure}[h]
    \centering
    \includegraphics[width=0.88\columnwidth]{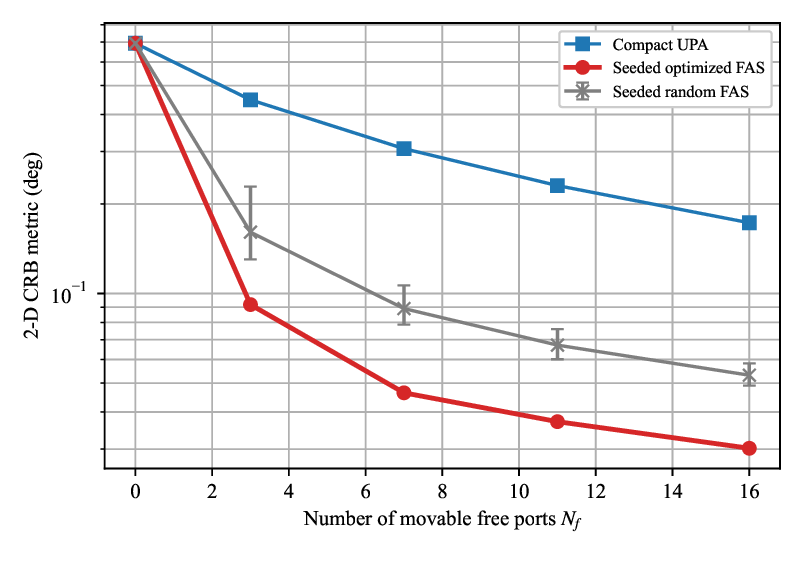}
    \caption{
    2-D CRB versus the number of movable free ports $N_f$, with the compact $3\times3$ seed fixed for the seeded random and optimized FAS designs.}
    \label{fig:n_ports}
\end{figure}

Figure~\ref{fig:movable_count} addresses whether the seed ports should also move. The CRB decreases as more ports are actuated because the feasible refinement set expands, numerically matching the monotonicity in Theorem~\ref{thm:finite_retention}. Each returned spacing- and movement-feasible layout supplies an $L_M$ certificate, whereas the reachable coordinate spans supply $U_M$ without assuming corner collocation. The numerically optimized all-movable design gives the lowest CRB in this comparison, while requiring every port to be actuated, switched, recalibrated, and settled after the seed observation. 
The plotted all-movable point is a feasible numerical benchmark; only the optimum over the larger feasible set defines the information upper bound.

\begin{figure}[!t]
    \centering
    \includegraphics[width=0.88\columnwidth]{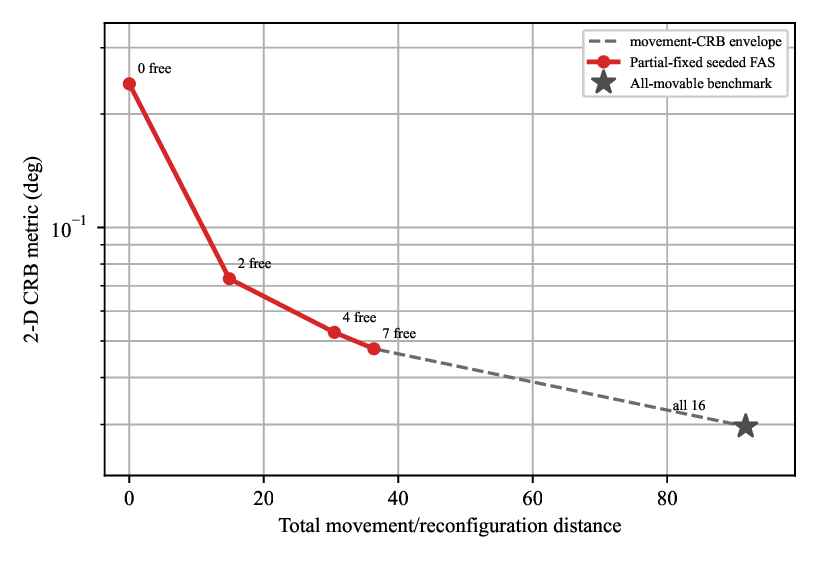}
    \caption{2-D CRB versus the total reconfiguration distance, with the partial-mobility envelope and all-movable benchmark.}
    \label{fig:movable_count}
\end{figure}

Figure~\ref{fig:seed_ablation} measures the probability that seed MUSIC resolves both sources within a one-degree 2-D RMSE. The $2\times2$ seed remains aperture limited, whereas the $4\times4$ seed crosses the resolution threshold at a lower SNR but consumes all 16 ports. The $3\times3$ seed provides the intermediate operating point: its resolution probability rises sharply near the target SNR while seven ports remain available for refinement movement. The estimator does not require seed-stage final accuracy; this stricter ablation quantifies the hardware cost of improving coarse resolution.

\begin{figure}[!t]
    \centering
    \includegraphics[width=0.88\columnwidth]{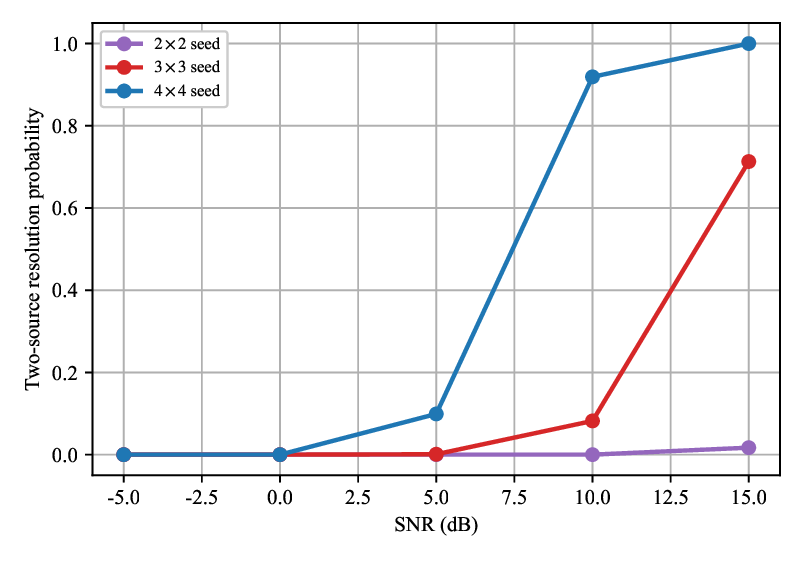}
    \caption{Two-source resolution probability of seed MUSIC versus the SNR for $2\times2$, $3\times3$, and $4\times4$ seed patches.}
    \label{fig:seed_ablation}
\end{figure}

Figure~\ref{fig:regularizer_ablation} separates aperture information from local coarray support. A D-optimal-only boundary design attains a low CRB but covers fewer rounded Nyquist lags, whereas the compact UPA has strong local coverage and a much smaller physical aperture. The proposed regularized design sacrifices a small amount of determinant gain to retain local lag support.
Hence, the rounded-lag metric shapes the physical free-port locations to promote local Nyquist structure, while the final estimator uses the exact manifold of the reconfigured array.
This ablation justifies the composite placement objective without claiming that the coarray score alone predicts estimator RMSE.

\begin{figure}[!t]
    \centering
    \includegraphics[width=0.88\columnwidth]{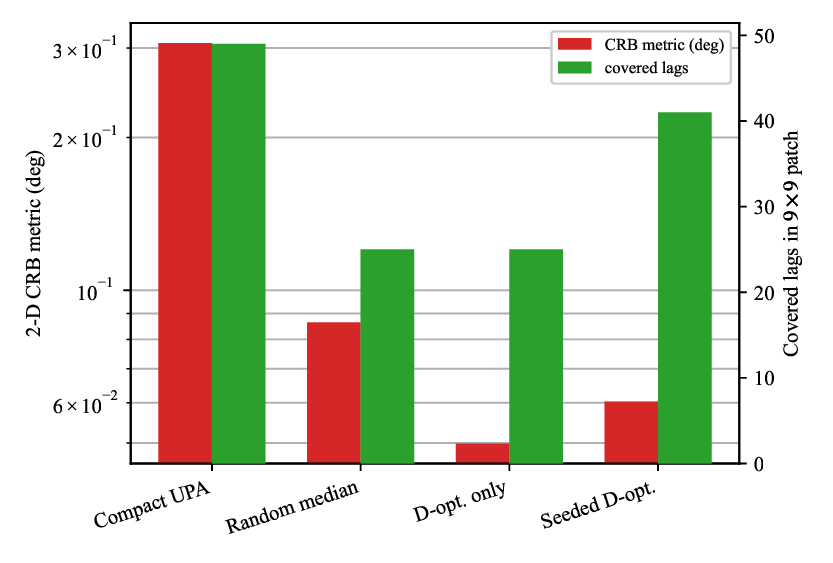}
    \caption{Placement-objective ablation of 2-D CRB and rounded-lag coverage for four array designs.}
    \label{fig:regularizer_ablation}
\end{figure}

\subsection{Optimization Behavior and Computational Cost}

Figure~\ref{fig:search_scaling} compares normalized operation counts as the 2-D search grid is refined. The seeded estimator remains close to the cost of the seed scan because it evaluates the joint likelihood only for a small gated candidate set; a global 2-D ML enumeration grows much faster. 
After normalization by the seed-scan cost at each grid size, the seed-scan curve remains at unity, while the seeded FAS-ML ratio approaches unity as the grid is  refined because its candidate-evaluation and local-refinement overhead grows much more slowly than the seed-scan cost. 
These curves support the complexity orders in Remark~\ref{rem:layout_complexity} and the estimator-complexity remark, while keeping optimization accuracy and search cost conceptually separate.



\begin{figure}[!t]
	\centering
	\includegraphics[width=0.88\columnwidth]{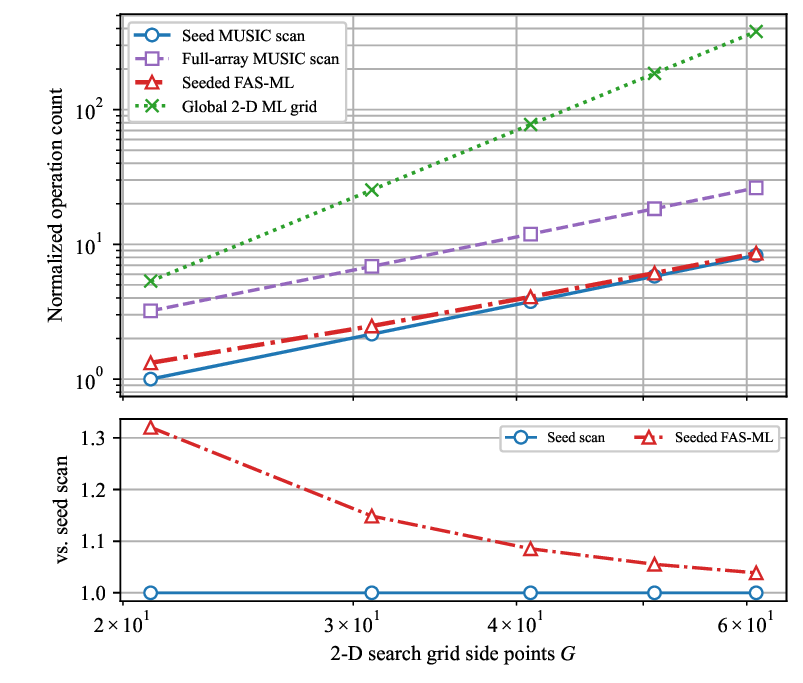}
	\caption{Normalized operation counts versus the 2-D search grid resolution, with complexity ratios relative to the seed scan.}
	\label{fig:search_scaling}
\end{figure}

\subsection{Robustness and Sector Coverage}

Figure~\ref{fig:robustness} generates refinement data at perturbed free-port coordinates while the estimator retains the nominal calibrated manifold. Both RMSE and upper-tail error increase with $\sigma_p$ because the unmodeled phase offsets accumulate across the moved baselines. The result measures model-mismatch sensitivity of the nominal design; it does not claim that the robust objective in Section~\ref{sec:robust} has been optimized.

\begin{figure}[!t]
    \centering
    \includegraphics[width=0.90\columnwidth]{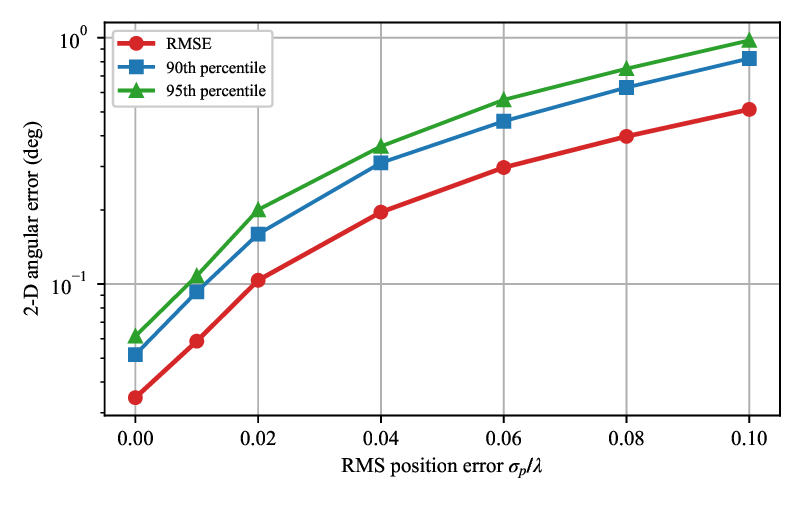}
    \caption{2-D angular error versus the normalized free-port position error.}
    \label{fig:robustness}
\end{figure}

Figure~\ref{fig:crb_heatmap} maps the CRB gain of the optimized FAS over the compact UPA while one source is fixed and the other moves across the sector. The gain remains positive over the tested region, but it is not uniform. As expected, the multi-source FIM becomes less well conditioned when the moving source approaches the fixed source or aligns with less favorable baselines. The result shows that the gain is not limited to one selected source pair, while also making clear that every finite-aperture DOA design remains source-geometry dependent.

\begin{figure}[!t]
    \centering
    \includegraphics[width=0.90\columnwidth]{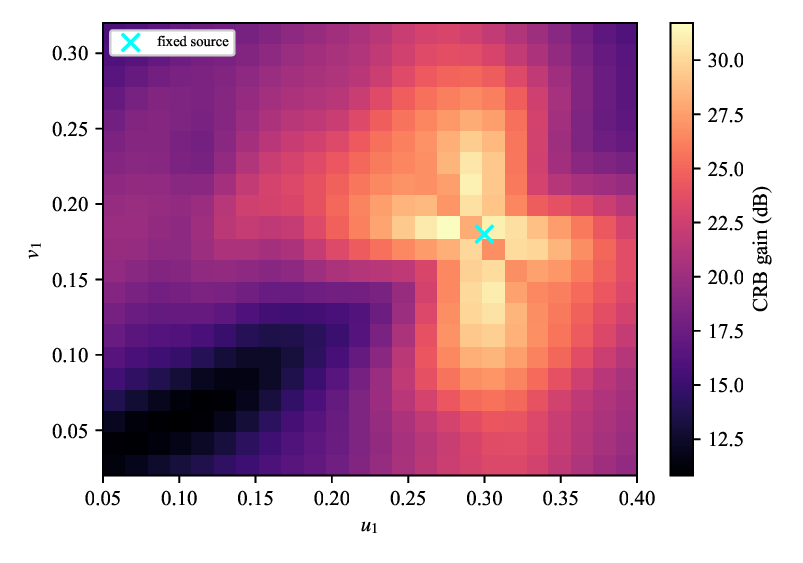}
    \caption{Sector-wide CRB gain of optimized FAS over compact UPA.}
    \label{fig:crb_heatmap}
\end{figure}

\subsection{Prior Adaptation}

The optimized free-port layouts remain prior dependent: broad sectors favor corner and edge diversity, whereas low-elevation or close-source priors favor baselines aligned with their dominant uncertainty directions.

Figure~\ref{fig:prior_shift} compares FAS with a permanently deployed continuous array. The frozen array is optimized for the broad-sector prior, remains fixed, and uses all $T_{\rm tot}$ snapshots; its embedded seed initializes single-state ML. Although the FAS reoptimizes the free ports after each prior shift, all frozen-to-reconfigured ratios remain below one. The permanent full-aperture design is sufficiently well conditioned that avoiding snapshot splitting outweighs adaptation. This result defines an important boundary: the proposed FAS is most relevant to compact standby, retractable deployment, or constrained permanent footprints rather than unconditional dominance over a full-aperture array.

\begin{figure}[!t]
    \centering
    \includegraphics[width=0.90\columnwidth]{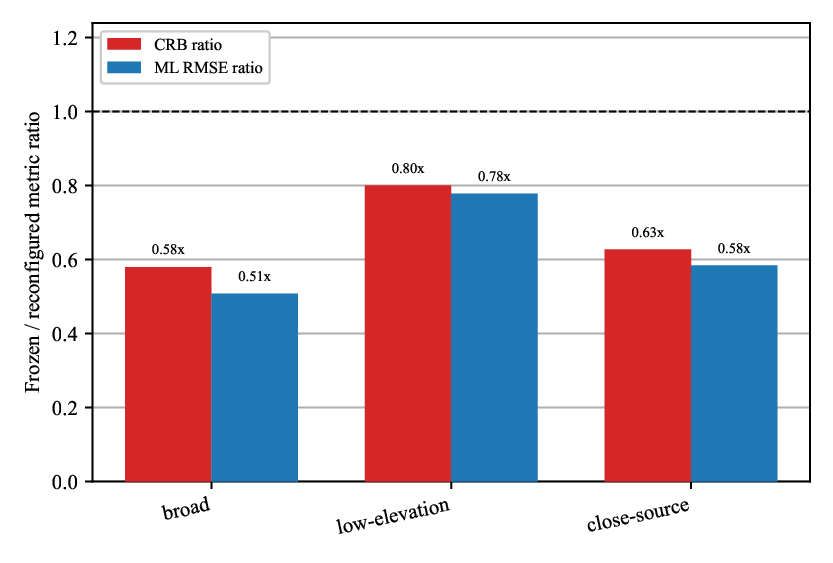}
    \caption{Frozen-to-reconfigured ratios under three prior shifts.}
    \label{fig:prior_shift}
\end{figure}


The experiments show that joint two-state estimation becomes increasingly
valuable for closely spaced sources, whereas reconfiguration provides limited
benefit when a full sparse aperture can remain permanently deployed.

\section{Conclusion}\label{sec:conc}

Starting from the limitations of static sparse arrays, this paper developed a two-state reconfigurable sparse aperture using a partial-mobility planar FAS. A compact, fixed Nyquist seed provides ambiguity-controlled acquisition, while selected fluid ports form source-conditioned sparse baselines for refinement. The policy-level FIM accounts for observation-dependent movement, while the single-source aperture law and finite-port certificate characterize the relaxed all-movable upper bound and feasible partial-mobility designs. An aperture-and-conditioning surrogate generates candidate layouts, which are ranked and locally refined by the exact two-state FIM, followed by seed MUSIC and joint concentrated ML for DOA estimation. Under equal snapshot budgets, simulations show that reconfiguration reduces angular error and retains much of the all-movable benefit with fewer actuated ports. This benefit is mainly limited by position mismatch, movement latency, and the availability of a permanently deployed full aperture. Partial-mobility FAS is therefore most useful when compact standby, retractable deployment, or limited actuation is required. Future work will consider joint movement and observation allocation and hardware validation under wideband, near-field, and coupling effects.

\end{document}